\documentclass[conference]{IEEEtran}

\usepackage[dvips]{color}
\usepackage{epsf}
\usepackage{times}
\usepackage{epsfig}
\usepackage{graphicx}
\usepackage{amsmath}
\usepackage{amssymb}
\usepackage{amsxtra}
\usepackage{here}
\usepackage{rawfonts}
\usepackage{times}
\usepackage{url}
\usepackage{cite}
\usepackage{float}
\newtheorem{theorem}{\bf Theorem}

\newtheorem{lemma}{\bf Lemma}

\newtheorem{definition}{\bf Definition}
\newtheorem{remark}{Remark}

\DeclareMathOperator*{\argmax}{arg\,max}

\newlength{\aligntop}
\newlength{\alignbot}
\makeatletter
\renewenvironment{align}{%
  \vspace{\aligntop}
  \start@align\@ne\st@rredfalse\m@ne
}{%
  \math@cr \black@\totwidth@
  \egroup
  \ifingather@
    \restorealignstate@
    \egroup
    \nonumber
    \ifnum0=`{\fi\iffalse}\fi
  \else
    $$%
  \fi
  \ignorespacesafterend%
  \vspace{\alignbot}\par\noindent
}

\IEEEoverridecommandlockouts
\begin{document}
\title{\huge A\! Game-Theoretic\! Approach\! to\! Energy\! Trading\! in\! the\! Smart\! Grid\vspace{-0.4cm}}
\author{\authorblockN{Yunpeng Wang$^\textbf{1}$, Walid Saad$^2$, Zhu Han$^3$, H. Vincent Poor$^4$, and Tamer Ba\c{s}ar$^5$}
\authorblockA{\small
$^\textbf{1,2}$ Electrical and Computer Engineering Department, University of Miami, Coral Gables, FL 33146\\Emails: \url{y.wang68@umiami.edu,walid@miami.edu}\\
$^3$ Electrical and Computer Engineering Department, University of Houston, Houston, TX, USA, Email: \url{zhan2@uh.edu}\\
$^4$ Electrical Engineering Department, Princeton University, Princeton, NJ, USA, Email: \url{poor@princeton.edu}\\
$^5$ Coordinated Science Laboratory, University of Illinois at Urbana-Champaign, IL, USA, Email: \url{basar1@illinois.edu}\vspace{-1cm}
 }%
\thanks{This research was supported by CenterPoint, US NSF CNS-1265268, and CNS-0953377, the Air Force Office of Scientific Research under MURI Grant FA9550-09-1-0643 and  the U.S. Air Force Office of Scientific Research (AFOSR) MURI grant FA9550-10-1-0573. A preliminary version of this work was presented in [35].}
 }
\date{}
\maketitle

\begin{abstract}
Electric storage units constitute a key element in the emerging smart grid system. In this paper, the interactions and energy trading decisions of a number of geographically distributed storage units are studied using a novel framework based on game theory. In particular, a noncooperative game is formulated between  storage units, such as PHEVs, or an array of batteries that are trading their stored energy. Here, each storage unit's owner can decide on the maximum amount of energy to sell in a local market so as to maximize a utility that reflects the tradeoff between the revenues from energy trading and the accompanying costs. Then in this energy exchange market between the storage units and the smart grid elements, the price at which energy is traded is determined via an auction mechanism. The game is shown to admit at least one Nash equilibrium and a novel proposed algorithm that is guaranteed to reach such an equilibrium point is proposed. Simulation results show that the proposed approach yields significant performance improvements, in terms of the average utility per storage unit, reaching up to $130.2\%$ compared to a conventional greedy approach. \vspace{-0cm}
\end{abstract}
\vspace{-0.3cm}

\section{Introduction}\vspace{-0.1cm}

Modernizing the electric power grid and realizing the vision of a ``smart grid'' is contingent upon the deployment of novel smart grid elements such as renewable energy sources and energy storage units~\cite{farhangi2010path}. In this respect, electric storage units are inherently devices that can store energy, or extra electricity available at participating customers. Deployment of storage units in future smart grid systems faces many challenges at different levels such as studying the impact of integrating storage units on the grid's operation, determining the required grid infrastructure (communication and control nodes) to enable smart energy exchange, and developing new power management strategies ~\cite{kazempour2009electric,hadjipaschalis2009overview, HosseinAkhavanHejazi, aguero2012integration, hossain2012smart}. The potential economic impact of deploying energy storage units was explored in \cite{lassila2012methodology}, which also studied the feasible level of energy storage in the distribution system. The possibility of having groups of controllable loads and sources of energy in power systems was investigated in \cite{hatziargyriou2007overview} which deployed a distribution network of solar panels or wind turbines. In \cite{thatte2012towards}, distributed resources are allocated by the provision of two-way energy flow and a unified, operational value proposition of energy storage is presented. Other related problems have assessed the advantages of deploying and maintaining storage units such as \cite{EW07, lindley2010smart, rastler2010electricity, lu2004pumped, garcia2008stochastic, sioshansi2009estimating, caralis2012role, diaz2012review, garnier2009integrated}.

One main challenge pertaining to introducing energy storage units within the smart grid is the analysis of the energy trading decision making processes involving complex interactions between the storage units (and their owners) and the various smart grid elements. A game theoretic approach to control individual sources/loads was adopted in  \cite{weaver2009game}, which enhanced the reliability and robustness of a power system without using central control. In \cite{EW02}, a new technique based on cooperative game theory is proposed to allow wind turbines to aggregate their generated energy and optimize their profits. Reactive power compensation was studied in \cite{xu2010research} using game theory with the objective of optimizing wind farm generation.  A strategic game model was developed in \cite{hobbs2000strategic} to analyze an oligopoly within an energy market with various grid-level constraints. Transmission expansion planning and generation expansion planning were studied through a dominant strategy using three game-theoretic levels in \cite{ng2006game}. Using an IEEE 30-bus test systems, a comprehensive approach to evaluate  electricity markets was presented in \cite{bompard2006network} to study the impact of various constraints on the market equilibrium. Developing a distributed energy storage system for transferring photovoltaic power to electric vehicles as well as introducing efficient power management schemes between storage units and the smart grid have been studied in \cite{aguero2012integration} and \cite{EW07}, respectively. However, little work seems to have been conducted, \emph{from the storage units' point of view}, on the energy exchange markets that arise due to the competition among a number of storage units, each of which could belong to a different customer and that can interact at different levels. Due to the promising outlook of introducing energy storage units in the smart grid, devising new schemes to model and analyze the competition accompanying such energy exchange markets is both challenging and desirable. 

The main contribution of this paper is to develop a new framework that enables a number of storage units belonging to different customers to individually and strategically choose the amount of stored energy that they wish to sell to customers in need of energy (e.g., other nodes or substations on the grid). Compared to related works on smart grid markets~\cite{EW07, lu2004pumped, bompard2006network, PH01, EW02}, our paper has several new contributions: \emph{1)} we design a novel double-auction market model that allows to incorporate power markets with multiple buyers and multiple sellers; \emph{2)} in contrast to the classical single shot, static auction models which assume that sellers have a constant amount to sell, we have developed here a novel framework that combines a double auction with a noncooperative game allowing the sellers to strategically decide on the amount they put for sale depending on the current market state, thus, yielding a dynamic pricing mechanism; \emph{3)} we have developed new results on the existence of a Nash equilibrium for games that exhibit a discontinuity in the utility function due to the presence of an underlying auction model, unlike the classical models that often assume continuous utilities, and \emph{4)} we proposed a new learning algorithm that is guaranteed to reach an equilibrium for a game with two levels of interactions: a market based on auction theory and a noncooperative game. We are particularly interested in overcoming two key challenges: (a) introducing a new approach using which the storage units can smartly decide on the energy amount to sell while taking into account the effect of these decisions on both their utilities and the energy trading price in the market, and (b) developing and analyzing a mechanism to characterize the trading price of the energy trading market that involves the storage units and the potential energy buyers in the grid. To this end, we model the competition between a number of storage units that are seeking to sell their surplus of energy as a noncooperative game in which the strategy of each unit is to select a self-profitable amount of energy surplus to sell so as to optimize a utility function that captures the tradeoff between the economical benefits of trading energy and the related costs (e.g., battery life reduction, storage unit efficiency, or other practical aspects). Then, a double auction mechanism is proposed to determine the trading price that potentially governs the market resulting from the interactions between the storage units and energy buyers. This mechanism is shown to be strategy-proof such that each buyer or seller has an incentive to be truthful in its reservation bids or prices. For the studied game, we show the existence of at least one Nash equilibrium and we propose a novel algorithm to find a Nash equilibrium of the game. Subsequently, we also show the convergence of the proposed algorithm. Extensive simulations are run to evaluate and assess the performance of the proposed game-theoretic approach.

The remainder of the paper is organized as follows: Section~\ref{sec:sysmodel} presents the studied system model trading. In Section~\ref{sec:game}, we formulate the game and develop the underlying auction mechanism. In Section~\ref{sec:algo}, we introduce the concept of the best response and describe our proposed algorithm. Simulation results are presented in Section~\ref{sec:sim}, while conclusions are drawn in Section~\ref{sec:conc}.\vspace{-0.2cm}

\section{System Model}\label{sec:sysmodel}\vspace{-0.1cm}
Consider a smart grid system having a number of nodes that are in need of energy. These nodes could represent substations and/or distributed energy sources that are servicing an area or group of consumers (e.g., loads, pumped-storage in hydro plants). Here, we consider that a certain number, $K$, of the smart grid elements is unable to meet their demand due to factors such as intermittent generation and varying consumption levels at the grid's loads. In this respect, such $K$ grid elements must find alternative sources of energy by acquiring this energy from other elements that have an excess of energy stored in an energy storage unit. Thus, we consider that a number, $N$, of storage units are deployed in the grid. In particular, all these $N$ units belong to customers that have an excess of energy that they wish to sell. We let $\mathcal{N}$ and $\mathcal{K}$ denote, respectively, the sets of all $N$ sellers and all $K$ buyers. In what follows, we use \emph{seller} to imply any storage unit $i\in \mathcal{N}$ and \emph{buyer} to imply any smart grid element $k \in \mathcal{K}$. Our model generally involves several types of electricity sellers and buyers.

Each buyer $k \in \mathcal{K}$ has a maximum unit price or reservation bid $b_k$ at which it is willing to participate in an energy trade with a seller. Since we focus on the storage units' perspective of the market, we assume that the buyers wish to buy a fixed amount of energy $x_k$. This models a scenario in which, over a certain given time period, $x_k$ is imposed on the buyers from the practical energy requirements of the users and customers. We can also view $x_k$ as an \emph{average} value of the amount of additional energy demand that buyer $k$ foresees for a certain period of time. For the storage units, i.e., the sellers, each unit $i \in \mathcal{N}$ can chose an amount of energy $a_i$ to sell such that:
 \begin{align}
a_i \le B_i \triangleq (C_{i,\textrm{max}} - D_i)
 \end{align}
with $B_i$ being the maximum total energy that seller $i$ wants to sell in the market, $C_{i,\textrm{max}}$ being the maximum storage unit capacity, and $D_i$ being the energy that each storage unit $i$ wants to keep and is not interested in selling. For each seller $i$, we define a reservation price $s_i$ per unit energy sold, under which seller $i$ will not trade energy. 
\begin{figure}[!t]
  \begin{center}
   \vspace{-0.2cm}
    \includegraphics[width=7cm]{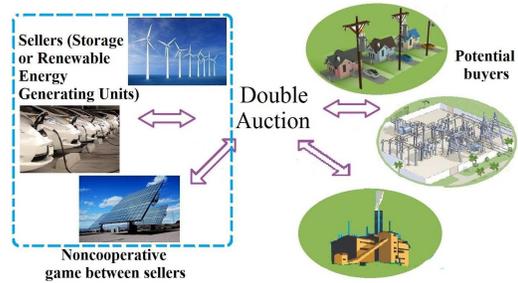}
    \vspace{-0.4cm}
    \caption{\label{fig:fig0} An illustrative example of the model studied.}
  \end{center}\vspace{-1cm}
\end{figure}

Given these buying and selling profiles of the various grid elements, an energy exchange market is set up in which the buyers seek to acquire energy so as to meet their demand while the sellers, i.e., the storage units and their owners, seek to collect revenues from selling their extra energy surplus. Here, all $N$ sellers and $K$ buyers will interact so as to determine various energy trading properties that include the quantities exchanged and the price at which energy is traded. Unlike conventional markets in which the sellers only control the reservation prices, in our model, the storage units can also strategically choose the maximum quantity of energy $a_i$ that they want to put for sale in the market. The choice of a proper $a_i$ is directly dependent on an inherent tradeoff between the potential profits that the sellers foresees and the accompanying costs that relate to the physical characteristics of the storage devices. Indeed, such a tradeoff is a byproduct of the fact that frequently charging or discharging storage devices is costly as it can lead to a reduction in the storage device's lifespan as well as to other practical costs~\cite{PH00,PH01,hossain2012smart}. Hence, given the buyers' set of bids and energy requirements, the maximum energy amount $a_i$ that any seller $i$ decides to trade strongly affects both the gains/revenues and cost of every storage unit in $\mathcal{N}$ as well as the trading price. Fig.~\ref{fig:fig0} provides an illustrative example of the model  considered. To analyze such an energy exchange market, we next propose a new framework that builds on the powerful analytical tools of game theory and auction theory. \vspace{-0.1cm}

\section{A Game-Theoretic Approach to Energy Trading}\label{sec:game}\vspace{-0.1cm}
In this section, we first formulate a noncooperative game between the sellers, and then study the proposed energy trading mechanism using a double auction, also discussing its various properties. The main notation is listed in Table~\ref{tab:symbol}.

\begin{table}[!t]\vspace{-0.2cm}
%\scriptsize
\small
  \centering
 \caption{%\mycaption{%\vspace*{-1em}
    \vspace*{-0.1cm}Summary of Notations}\vspace{-0.4cm}%\vspace{-0.6em}
\begin{tabular} {c||c}
Symbols & Description\\ \hline\hline %[0.5ex]
$N$ & the total number of sellers\\
$K$ & the total number of buyers\\
$i,j$ & the serial number of sellers\\
$k$ & the serial number of buyers\\
$s_i$ & the reservation price of seller $i$\\
$b_k$ & the reservation bid of buyer $k$\\
$a_i$ & the action of seller $i$\\
$x_k$ & the demand of buyer $k$\\
$L$ & the total number of participating sellers\\
$M$ & the total number of participating buyers \\
$U_i$ & the utility of seller $i$\\
$\boldsymbol{a}_{-i}$ & the sellers' actions except seller $i$\\
$p$ & the trading price\\
$q$ &the sold energy \\
$Q$ & the energy exchange function of action\\
$\beta$ &the oversupply energy \\
\end{tabular}\label{tab:symbol}\vspace{-0.6cm}
\end{table}

\subsection{Noncooperative Game Model}\vspace{-0.1cm}
The complex interactions and decision making processes of the storage units are analyzed using the analytical tools of noncooperative game theory~\cite{GT00}.
In particular, we formulate a noncooperative game in normal form, $\Xi=\{\mathcal{N}, \{\mathcal{A}_i\}_{i\in\mathcal{N}}, \{U_i\}_{i\in\mathcal{N}}\}$, that is characterized by three main elements: (a) a set $\mathcal{N}$ of sellers or \emph{players}, (b) action or \emph{strategy} of each player $i \in \mathcal{N}$ which maps to an amount of energy, $a_i \in \mathcal{A}_i := [0,B_i]$, that will be sold, and (c) a \emph{utility function} $U_i$ of each seller $i \in \mathcal{N}$ which reflects the gains and costs from trading and selling energy. Before defining the utility functions, we note that, in the game $\Xi$, the reservation price $s_i$ is not included as part of seller $i$'s strategy space.  This implies that the sellers must reveal their correct reservation price when participating in the game. This consideration is motivated by the fact that, when we determine the market's trading price, as explained in the next section, we will develop a \emph{truthful} and strategy-proof double auction mechanism that guarantees that no buyer or seller can benefit by cheating or changing its true reservation price or bid.

Given a certain strategy choice $a_i$ by any storage unit $i\in \mathcal{N}$, the utility function can be characterized by:
\begin{equation}\label{eq:util1}
U_i(a_i,\boldsymbol{a}_{-i})=  \sum_{k \in \mathcal{K}} (p_{ik}(\boldsymbol{a})-s_i)q_{ik}(\boldsymbol{a}) -  f\left(\sum_{k \in \mathcal{K}}q_{ik}(\boldsymbol{a}) \right),
\end{equation}
where $\boldsymbol{a}$ is the $N\times 1$ vector of all strategy selections, $\boldsymbol{a}_{-i}:=[a_1,a_2,\ldots,a_{i-1},a_{i+1},\ldots,a_N]^T$ is the vector of actions selected by the opponents of storage unit $i$, $p_{ik}(\boldsymbol{a})$ is the price at which energy is traded between seller $i$ and buyer $k$,  $q_{ik}$ is the quantity of energy exchanged from seller $i$ to buyer $k$, and $f(\cdot)$ is a function that reflects the cost of selling energy. As previously mentioned, these costs depend on numerous factors such as the physical type of the storage unit or the amount of time the unit is put into charging or discharging modes. Moreover, we note that $f$ must be an increasing function of the amount, $\sum_{k \in \mathcal{K}}q_{ik}(\boldsymbol{a})$, sold in total by storage unit $i$. Here, the utility in (\ref{eq:util1}) is also a function of the amount $x_k$ of energy that every buyer $k\in \mathcal{K}$ must buy and of the buyers' reservation bids. However, for notational convenience, we have dropped this dependence.

The goal of each storage unit $i$ is to choose a strategy $a_i \in \mathcal{A}_i$ in order to maximize its utility as given in (\ref{eq:util1}). For characterizing a desirable outcome for the studied game $\Xi$, one must derive a suitable solution for all $N$ optimization problems that the sellers need to solve. We can first see that, in (\ref{eq:util1}), every vector of strategies $\boldsymbol{a}$ selected by the sellers will yield different trading prices $p_{ik}(\boldsymbol{a})$. These prices are also a function of the reservation prices of the sellers, the quantity bought, and the reservation bids of the buyers or grid elements. Thus, prior to finding a solution for the energy exchange game, we will first introduce a scheme for characterizing the trading price.\vspace{-0.1cm}

\subsection{Double Auction Mechanism for Market Analysis}\label{sec:da}\vspace{-0.1cm}
The formulated game is useful to study the sellers' interactions. However, in order to find the prices at which energy is traded, we must define suitable mechanisms using the rich tools of double auctions~\cite{DR00} and \cite{DR01}. Inherently, a double auction is a suitable representation for a trading market that involves multiple sellers and multiple buyers. For the proposed game $\Xi$, applying a double auction is needed so as to derive the trading prices, the quantities of energy traded, as well as the number of involved sellers and buyers, given the chosen strategy vector $\boldsymbol{a}$ (maximum quantities offered for sale), the reservation prices $s_i,\ \forall i \in \mathcal{N}$, the quantities $x_k$ to be bought, and the bids $b_k,\ \forall k \in\mathcal{K}$. 

When dealing with a double auction, the buyers and sellers have to decide on whether to be truthful about their reservation bids and prices, given: (i) the potential utility that they will obtain as captured by the first term of (\ref{eq:util1}), and (ii) the buyers' potential savings $\sum_{i \in \mathcal{N}} (b_k - p_{ik})q_{ik}$ with $q_{ik}$ being the quantity bought by $k$ from $i$. Here, our emphasis is on having a double auction mechanism that yields, for any  $\boldsymbol{a}$, a solution that is \emph{truthful and strategy-proof}. A truthful auction is a scheme where no seller $i \in \mathcal{N}$ can benefit by cheating about its reservation price such as by misreporting it to $s_i^\prime > s_i$ or $s_i^\prime < s_i$, and no buyer $k\in \mathcal{K}$ will gain by under-bidding $b_i^\prime < b_i$ or over-bidding $b_i^\prime > b_i$. A strategy-proof solution is of interest as it guarantees truthful reporting by all buyers and sellers.

With this in mind, for the proposed system, we develop a double auction scheme that follows from \cite{DR00} and \cite{DR01}. In this scheme, the first step is to sort the sellers in an \emph{increasing order} of their reservation prices such that, without loss of generality, we have:
\begin{align}\label{eq:sel}
s_1 < s_2 < \ldots < s_N.
\end{align}
The next step is to arrange the buyers in a decreasing order of their reservation bids, as follows:
\begin{align}\label{eq:buy}
b_1  > b_2 > \ldots > b_K.
\end{align}
We note that the orderings in (\ref{eq:sel}) and (\ref{eq:buy}) assume that whenever two buyers or sellers have equal reservation prices or bids, one can group them into a single, virtual buyer or seller.

Following this sorting process, the \emph{supply} curve (sellers' price $s_i$ as a function of the energy amount $a_i$, $\forall i\in \mathcal{N}$)  and the \emph{demand} curve (buyers' bids $b_k$ as a function of the amount of required energy $x_k$, for all $k\in \mathcal{K}$ can be generated. These two curves will subsequently \emph{intersect} at a point that corresponds to a given seller $L$ and a certain buyer $M$ with $b_M \ge s_L$. This intersection point is easily computed using known numerical and graphical techniques~\cite{DR00}. Once we determine the seller $L$ and buyer $M$ at the supply and demand intersection point, double auction theory implies that $L-1$ and $M-1$ buyers will practically participate in the market and the energy trading process. Here, as shown in \cite{DR01}, we must exclude seller $L$ and buyer $M$ from the market so as to guarantee that the total supply and demand will match while ensuring a strategy proof and truthful auction mechanism. However, if one does not need to maintain truthfulness, the proposed approach can easily be modified so as to allow seller $L$ and buyer $M$ to also trade energy.

Therefore, in order to match the supply and demand, all sellers whose indices are such that $i < L$ and all buyers such that $k < M$ will be part of the double auction trade. To determine the trading price, once the intersection is identified, one can select any suitable point within the interval $[s_L,b_M]$~\cite{DR00}. For our energy market, given a seller's strategy vector $\boldsymbol{a}$, we assume that all sellers $i < L$ and buyers $k<M$ will exchange energy at a price $\bar{p}(\boldsymbol{a})$ such that:
\begin{align}\label{eq:tradpr}
\bar{p}(\boldsymbol{a})=\frac{s_L + b_M}{2}.
\end{align}
Here, the price depends on $\boldsymbol{a}$ since, for every maximum energy to sell vector $\boldsymbol{a}$, the intersection point of demand and supply may occur at different $M$ and $L$.
\begin{figure}[!t]
\begin{center}\vspace{-0.2cm}
\includegraphics[width=7cm]{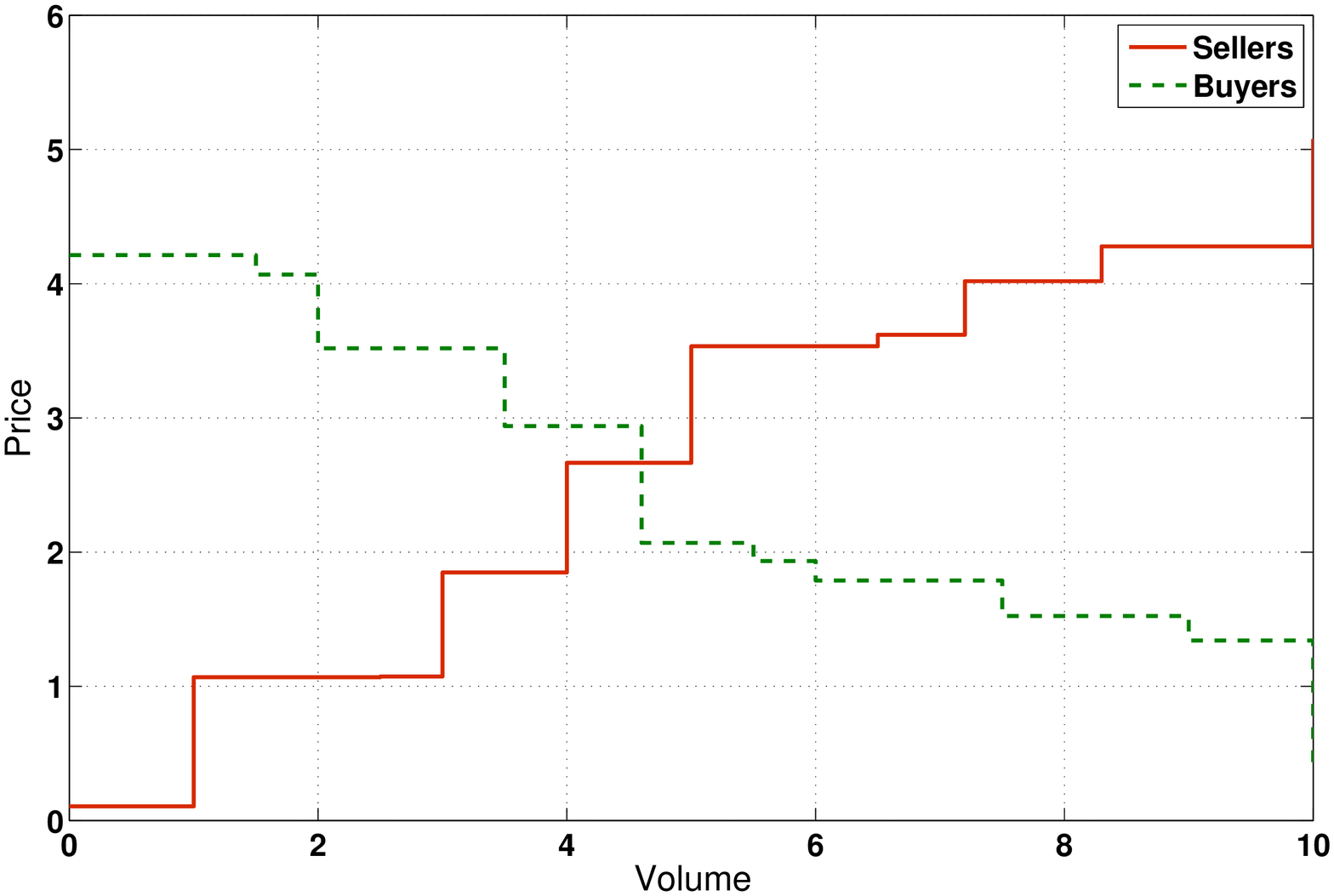}
\end{center}\vspace{-0.7cm}
\caption{An illustrative example on solving a double auction }\label{fig:intersect} \vspace{-0.8cm}
\end{figure}

This solution of the double auction is illustrated in Fig.~\ref{fig:intersect}, which shows the supply (solid line) and demand (dashed line) curves. The intersection of these two curves can be used to determine the trading price and the quantities. Once the trading price and quantities are determined from the double auction, the next step is to define a proper utility function and introduce the strategic operation for the proposed game.

Once the trading price is found, we need to find the amount of energy that is traded between the $L-1$ sellers and the $M-1$ buyers. First, given the unified trading price in (\ref{eq:tradpr}), the $L-1$ sellers will be indifferent between buyers. This implies that, at the double auction solution, each seller's utility in (\ref{eq:util1}) depends only on the quantity sold but not on the identity of the buyer who bought this amount. Hence, assuming that the cost function $f(\cdot)$ is quadratic, by using the proposed double auction, (\ref{eq:util1}) becomes:
 \begin{align}\label{eq:util}
U_i(a_i,\boldsymbol{a}_{-i})=  (\bar{p}(\boldsymbol{a})-s_i)Q_{i}(\boldsymbol{a})  -  \tau_i Q_{i}^2(\boldsymbol{a}),
\end{align}
with $Q_i(\boldsymbol{a})$ being the \emph{total quantity} of energy sold by  $i$ and $\tau_i$ being a penalty factor that weighs the costs reaped by storage unit $i$ when discharging/selling energy. Here, we must stress that our analysis can accommodate any type of cost functions $f$.

Once the auction is concluded, different approaches can be applied to find the quantity of energy traded between each of the $L-1$ participating sellers and $M-1$ participating buyers~\cite{DR00}. For our work, we will apply the technique of \cite{DR01} where the entire volume traded is divided in a way to maintain the truthfulness of the auction. Using this approach, the total amount $Q_i(\boldsymbol{a})$ that is sold by any storage unit $i$, for a given strategy vector $\boldsymbol{a}$ is:
 \begin{equation}\label{eq:rule}
 Q_i(\boldsymbol{a})=\begin{cases}
a_i &\textrm{ if } \sum_{k=1}^{M-1} x_k \ge \sum_{j=1}^{L-1} a_j,\\
(a_i -  \beta_i)^+ &\textrm{ if }  \sum_{k=1}^{M-1} x_k \le \sum_{j=1}^{L-1} a_j,
\end{cases}
 \end{equation}
where $(\alpha)^+:= \max (0, \alpha)$ and $\beta_i$ represents the fraction of the oversupply $\sum_{j=1}^{L-1} a_j-\sum_{k=1}^{M-1} x_k$ that is allotted to seller $i$. The mechanism in (\ref{eq:rule}) implies that whenever the total demand at the auction's outcome exceeds the supply, then every seller $i$ would sell all of the energy $a_i$ that it introduced to the market. However, when the total supply exceeds the total demand, then all sellers get an \emph{equal} share of the oversupply's burden. Here, $\beta_i = \frac{\sum_{j=1}^{L-1} a_j-\sum_{k=1}^{M-1} x_k}{L-1}$. Nonetheless, if, for a seller $i$, we have $\frac{( \sum_{j=1}^{L-1} a_j-\sum_{k=1}^{M-1} x_k)}{L-1} > a_i$, then, seller $i$ does not sell any energy as per the second case in (\ref{eq:rule}). The remaining ``oversupply'' $\frac{( \sum_{j=1}^{L-1} a_j-\sum_{k=1}^{M-1} x_k)}{L-1} - a_i$ of this seller is subsequently divided equally between the other $L-2$ sellers and the result is added to their share $\beta_j,\ j < L,\ j \neq i$. This scheme will be repeated as long as each seller sells a nonnegative quantity. Here, for the ``oversupply" case, the total energy put in the market by the participating sellers, $\sum_{j=1}^{L-1} a_j$, is greater than or equal to that requested by the buyers, $\sum_{k=1}^{M-1} x_k$, but the real energy obtained by seller $k$ is the expected energy, and thus mathematically, $\sum_i q_{ik} = x_k$. For the ``over-demand" case, the amount of energy requested by the buyers exceeds the amount put into the market by the sellers, that is, $\sum_i q_{ik} < x_k$. An analogous process can be carried out to find the amount bought by the grid's elements or buyers.  Using (\ref{eq:rule}), as shown in \cite{DR00} and \cite{DR01}, we will have:
\begin{lemma}
 In the proposed game $\Xi$, by using (\ref{eq:rule}), no seller or buyer benefits by cheating about its reservation price $s_i,\ \forall i\in\mathcal{N}$ or reservation bid $b_k,\ \forall k\in\mathcal{K}$. The double auction is thus \emph{strategy-proof or truthful}.\vspace{-0.2cm}
 \end{lemma}

\section{Proposed Solution and Algorithm}\label{sec:algo}\vspace{-0.2cm}
Any storage unit $i \in \mathcal{N}$ can use the proposed double auction in order to estimate its utility, as per (\ref{eq:util}), for every $a_i$ given the strategy choices $\boldsymbol{a}_{-i}$ of its opposing players. Each seller seeks to maximize its utility by selecting the proper strategy $a_i \in \mathcal{A}_i$. In order to solve a noncooperative game in normal form such as $\Xi$ , one popular solution is that of a \emph{Nash equilibrium}~\cite{GT00}. A Nash equilibrium is a state of the game such that no player can increase its utility by \emph{unilaterally} deviating from this equilibrium state. Formally, the Nash equilibrium is defined as follows~\cite{GT00}:
\begin{definition}
Consider the proposed noncooperative game in normal form $\Xi=\{\mathcal{N}, \{\mathcal{A}_i\}_{i\in\mathcal{N}}, \{U_i\}_{i\in\mathcal{N}}\}$, with $U_i$ given by (\ref{eq:util}) given the underlying double auction. A vector of strategies $\boldsymbol{a}^*$ is said to be at a \emph{Nash equilibrium~(NE)}, if and only if, it satisfies the following set of inequalities:
\begin{align}
U_i(a_i^*, \boldsymbol{a}_{-i}^*) \geq U_i(a_i, \boldsymbol{a}_{-i}^*) , \ \ \forall a_i\in\mathcal{A}_i,\ i\in\mathcal{N}.
\end{align}
\end{definition}

Next, we first prove the existence of an NE for the proposed game and, then, we propose an algorithm that could find an NE for our model. Before going through our analysis, we first point out that, in general, the existence of an NE is not guaranteed for any noncooperative games. In particular, when the strategy spaces $\mathcal{A}_i,\ \forall i \in \mathcal{N}$ are compact such as in our proposed game $\Xi$, the existence of an NE is contingent upon having a utility function in (\ref{eq:util}) that is continuous in $a_i$~\cite{GT00}. However, in our game, the double auction process introduces a discontinuity in the utilities in (\ref{eq:util}) due to the dependence on the trading price. Showing the existence of NE for a discontinuous utility function is known to be more challenging than the classical case in which the utility is a continuous function~\cite{GT00}. Nonetheless, for the proposed game, we can obtain the following existence result:
\begin{theorem}\label{th:existence}
For the noncooperative game $\Xi=\{\mathcal{N},\{\mathcal{A}_i\}_{i\in\mathcal{N}}, \{U_i\}_{i\in\mathcal{N}}\}$, there exists at least one pure-strategy Nash equilibrium.
\end{theorem}
\begin{proof}
Assume $\mathcal{A}_i \subseteq \mathbb{R}^m (i=1, \dots, \mathcal{N})$, is a non-empty, convex and compact set. As shown in \cite{dasgupta1986existence}, if $\forall i$, $U_i: \mathcal{A} \rightarrow \mathbb{R}^1$ is \emph{1)} graph-continuous \emph{2)} upper semi-continuous in $a$ \emph{3)} quasi-concave in $a_i$, then the game $\Xi=\{\mathcal{N}, \{\mathcal{A}_i\}_{i\in\mathcal{N}}, \{U_i\}_{i\in\mathcal{N}}\}$ possesses a pure-strategy Nash equilibrium.

A function $U_i$ is said to be graph continuous if there exists a function $A_i=F_i(A_{-i}),\ \forall a \in A$ such that $U_i(F_i(a_{-i}),a_{-i})$ is continuous in $\boldsymbol{a}_{-i}$. In particular, as shown in \cite{dasgupta1986existence}, a piecewise continuous function is graph continuous if the strategy space is compact. Due to the discontinuity of our utility function as $\boldsymbol{a}_{-i}$ varies, we define
\begin{equation}
\begin{split}
\bar{p}(a_{i},\boldsymbol{a}_{-i})=p_1,\\
\bar{p}(a_{i},\boldsymbol{a}_{-i}+\Delta a)=p_2.\\
\end{split}
\end{equation}
at a jump point. Thus, we can assume that in a given range $\Delta a$,
\begin{equation}
\bar{p}(a_{i},\boldsymbol{a}_{-i})=\frac{p_2-p_1}{\Delta a}\Delta \boldsymbol{a}_{-i},
\end{equation}
and when $\Delta a$ is a small value, the slope becomes infinite, which implies that the function is a piecewise continuous function on $\boldsymbol{a}_{-i}$.

Mathematically, $U_i(a_{i},\boldsymbol{a}_{-i})$ is upper semi-continuous at $a_{i0}$ if there exists a neighborhood $a_{i}$ such that
\begin{equation}
\lim_{a_i \rightarrow a_{i0}} \sup U_i(a_{i},\boldsymbol{a}_{-i})\le U_i(a_{i0},\boldsymbol{a}_{-i}).
\end{equation}
Similarly, for a jump point of utility function, we define $a_{i0}=a_{i}+\Delta a$ such that $p_3 \le p_4$,
\begin{equation}
\begin{split}
\bar{p}(a_{i},\boldsymbol{a}_{-i})=p_3,\\
\bar{p}(a_{i0},\boldsymbol{a}_{-i})=p_4.\\
\end{split}
\end{equation}
Clearly, the utility function is upper semi-continuous because rational players seek a higher profit around the jump point. Here, we only need to prove that the utility function is quasi-concave.

In a given range of constant price, by simplifying (\ref{eq:util}) and (\ref{eq:rule}), we have:
\begin{equation}\label{eq:UQa}
\begin{split}
f(Q) &=(\bar p-s_i)Q -  Q^2, \\
Q(\boldsymbol{a}) &=\begin{cases}
a_i, & \text{if $\sum\nolimits_{k=1}^{M-1}x_k \geq \sum\nolimits_{j=1}^{L-1}a_j$},\\
(a_i-\beta_i)^+, & \text{if $\sum\nolimits_{k=1}^{M-1}x_k \leq \sum\nolimits_{j=1}^{L-1}a_j$},
\end{cases}\\
U_i &=f(Q(\boldsymbol{a})).
\end{split}
\end{equation}
Before proceeding further with the proof, we need to state the following Lemma from \cite{BO00}:

\begin{lemma} Suppose $g: X \rightarrow R$ is quasilinear and $h: g(X) \rightarrow R$ is a quasi-concave function. Then $h \circ g: X \rightarrow R$ is quasi-concave. 
\end{lemma}

This result has been extended to concave functions with strict conditions in \cite{BO00}. Now,
\begin{equation}\label{eq:Qa}
\begin{split}
\frac{\partial Q(\boldsymbol{a})}{\partial a_i} \ge& 0,\\
Q(\lambda a_{i}^x + (1-\lambda) a_{i}^y, \boldsymbol{a}_{-i}) \geq & \min[Q(a_{i}^x),Q(a_{i}^y), \boldsymbol{a}_{-i}],\\
Q(\lambda a_{i}^x + (1-\lambda) a_{i}^y, \boldsymbol{a}_{-i}) \leq & \max[Q(a_{i}^x),Q(a_{i}^y), \boldsymbol{a}_{-i}],\\
& \forall \ a_{i}^x\neq a_{i}^y, \lambda \in (0,1),\\
\end{split}
\end{equation}
where $a_{i}^x$ and $a_{i}^y$ belong to the action set $\mathcal{A}_i$ of seller $i$. Thus, $Q(\boldsymbol{a})$ is both quasi-concave and quasi-convex, and hence it is a quasi-linear function. Subsequently we can obtain the partial derivative with respect to $Q$ from (\ref{eq:UQa}):
\begin{equation}\label{eq:UQ}
\begin{split}
\frac{\partial U(Q)}{\partial Q} =& \bar p-s_i  -  2Q,\\
U(\lambda Q_x + (1-\lambda) Q_y) \geq & \lambda U(Q_x)+(1-\lambda) U(Q_y),\\
& \forall \ Q_x\neq Q_y, \lambda \in (0,1).\\
\end{split}
\end{equation}
Thus, $U(Q)$ is a concave function (which is quasi-concave).
Following Lemma 2, we substitute (\ref{eq:Qa}) into (\ref{eq:UQ}),
\begin{equation}
\begin{split}
&U[Q(\lambda a_{i}^x + (1-\lambda) a_{i}^y,\boldsymbol{a}_{-i})]\\
\geq& U[\min\{Q(a_{i}^x,\boldsymbol{a}_{-i}),Q(a_{i_y},,\boldsymbol{a}_{-i})\}],\\
\geq& \min\{U[Q(a_{i}^x,\boldsymbol{a}_{-i})],U[Q(a_{i_y},\boldsymbol{a}_{-i})]\}.
\end{split}
\end{equation}
Thus, $U$ is a quasi-concave function of $a_i$.

Intuitively, the partial derivative of $U$ on $Q$ is positive before the local maximum of $U$ and negative after it. The partial derivative of $Q$ on $a_i$ is $1$ or a nonnegative number depending on the auction except at the inflection point when $a_i=\beta_i$. Around this inflection point, $U(a_i,\boldsymbol{a_{-i}})$ firstly increases then might decrease as $a_i$ varies in a price-holding graph-continuous range. Therefore, $U(a)$ is a quasi-concave function in $a$, and the game $\Xi=\{\mathcal{N},\{\mathcal{A}_i\}_{i\in\mathcal{N}}, \{U_i\}_{i\in\mathcal{N}}\}$ possesses a pure-strategy Nash equilibrium as it satisfies all required conditions.
\end{proof}

At any NE of the proposed game, no storage unit can improve its utility by \emph{unilaterally} changing the maximum quantity of energy that it wishes to sell, given the equilibrium strategies of the other storage units. Having established existence, we must develop a scheme that allows to reach an NE of the game $\Xi$. To do so, we must first define the notion of a \emph{best response}:
\begin{definition}
The \emph{best response} $r(\boldsymbol{a}_{-i})$ of any storage unit $i \in \mathcal{N}$  to the vector of strategies $\boldsymbol{a}_{-i}$ is a set of strategies for seller $i$ such that:
\begin{align}\label{eq:br}
\small
r(\boldsymbol{a}_{-i})\!=\!\{a_i \in \mathcal{A}_i|U_i(a_i,\boldsymbol{a}_{-i}) \ge U_i(a_i^\prime,\boldsymbol{a}_{-i}),\ \forall a_i^\prime \in \mathcal{A}_i\}.
\end{align}
\end{definition}
Hence, for any storage unit $i \in \mathcal{N}$, when the other storage units' strategies are chosen as given by $\boldsymbol{a}_{-i}$, any best response strategy in $r(\boldsymbol{a}_{-i})$ is at least as good as any other strategy in $\mathcal{A}_i$. Using the concept of a best response, we can subsequently define a novel algorithm that can be used by the storage units and buyers so as to exchange energy. In particular, we propose the following iterative algorithm that is guaranteed to converge to a Nash equilibrium of the game:

\begin{theorem}\label{th:converge}
There exists a searching inertia weight $w$, $0<w<1$, such that, the iterative algorithm
\begin{equation}\label{eq:al}
a_i^{(n+1)}=(1-w) r(\boldsymbol{a}_{-i}^{(n)}) + w a_i^{(n)},
\end{equation}
converges to an NE.
\end{theorem}
\begin{proof}
In the proposed model, the classical best response dynamics may not converge due to the underlying auction mechanism. Depending on the different amounts between sellers and buyers, the price is a piecewise continuous function. From (\ref{eq:util}) and (\ref{eq:br}), we have:
\begin{equation}\label{eq:twocases}
\begin{split}
&r(\boldsymbol{a}_{-i})=\argmax_{a_i} [(\bar{p}(\boldsymbol{a})-s_i)Q_{i}(\boldsymbol{a})  -  \tau_i Q_{i}^2(\boldsymbol{a})],\\
=&\begin{cases}
\frac{\bar{p}(\boldsymbol{a})-s_i}{2\tau_i},\\
\qquad \text{if $\sum_{j=1}^{L-1} a_j^{(n)} \le \sum_{k=1}^{M-1} x_k$, (sell$\le$buy)};\\
\frac{(\bar{p}(\boldsymbol{a})-s_i)(L-1)+2\tau_i(\sum_{j=1,j\neq i}^{L-1}a_j-\sum_{k=1}^{M-1}x_k)}{2\tau_i(L-2)},\\
\qquad \text{if $\sum_{j=1}^{L-1} a_j^{(n)} \ge \sum_{k=1}^{M-1} x_k$, (sell$\ge$buy)}.\\
\end{cases}
\end{split}
\end{equation}
When selling amounts are less than buying, the best response of each seller is a constant. When selling amounts are greater than buying, and, we only use best response for iterations, we have
\begin{equation}\small
a_i^{(n+1)}=\frac{1}{L-2}\sum_{j=1,j\neq i}^{L-1}a_j^{(n)}+G_i,
\end{equation}
where $G_i=\frac{(\bar{p}(\boldsymbol{a})-s_i)(L-1)}{2\tau_i(L-2)}-\frac{1}{L-2}\sum_{k=1}^{M-1} x_k.$
To sum all $L-1$ sellers,
\begin{equation}\label{eq:brrange1}\small
\sum_{1}^{L-1}a_i^{(n+1)}=\sum_{1}^{L-1}a_i^{(n)}+\sum_{1}^{L-1}G_i.
\end{equation}
The second term on right hand side is not $0$ and this might lead a price change. In other words, in iteration $\gamma$,
\begin{equation}\label{eq:brrange2}\small
\begin{cases}
\textrm{if } \sum_{k=1}^{M-2} x_k \le \sum_{j=1}^{L-1} a_j^{(\gamma)} \le \sum_{k=1}^{M-1} x_k,\!\!\!\!&\bar{p}(\boldsymbol{a})=p_1\text{(sell$\le$buy)},\\
\textrm{if }  \sum_{k=1}^{M-1} x_k \le \sum_{j=1}^{L-1} a_j^{(\gamma)} \le \sum_{k=1}^{M} x_k,\!\!\!\!&\bar{p}(\boldsymbol{a})=p_2\text{(sell$\ge$buy)}.
\end{cases}
 \end{equation}
It is possible that, in iteration $\gamma+1$, the best response changes the total amount $\sum a_i^{(\gamma+1)}$ in (\ref{eq:brrange1}) and this can lead to a price changing loop.

However, from (\ref{eq:util}) and (\ref{eq:br}), our proposed algorithm in (\ref{eq:br}) has
\begin{equation}\small
a_i^{(n+1)}=wa_i^{(n)}+(1-w)\frac{1}{L-2}\sum_{j=1,j\neq i}^{L-1}a_j^{(n)}+(1-w)G_i.
\end{equation}
Similarly, to sum over all $L-1$ sellers,
\begin{equation}\label{eq:alrange1}\small
\sum_{1}^{L-1}a_i^{(n+1)}=\sum_{1}^{L-1}a_i^{(n)}+(1-w)\sum_{1}^{L-1}G_i.\\
\end{equation}
In particular, when $n>\gamma$, we have (\ref{eq:alrange1}). Thus, there must exist a weight $w$, such that
\begin{equation}\label{eq:alrange2}
\small
\begin{cases}
\textrm{if } \sum_{k=1}^{M-2} x_k \le \sum_{j=1}^{L-1} a_j^{(\gamma+1)} \le \sum_{k=1}^{M-1} x_k,\!\!\!\!&\bar{p}(\boldsymbol{a})=p_1\text{(sell$\le$buy)},\\
\textrm{if }  \sum_{k=1}^{M-1} x_k \le \sum_{j=1}^{L-1} a_j^{(\gamma+1)} \le \sum_{k=1}^{M} x_k,\!\!\!\!&\bar{p}(\boldsymbol{a})=p_2\text{(sell$\ge$buy)}.
\end{cases}
\end{equation}
\remark We note that the above result is generated for the case in which the $(L-1)$ sellers are able to sustain the oversupply. Similar results can easily be generated for the cases in which the oversupply cannot be split among all sellers; however, this case is omitted due to space limitations.
\remark With a given price range, the utility in (\ref{eq:util}) can be viewed as a concave function. More precisely, there exists a weight, such that no player could arbitrarily approach its current best response, which might change the price and pull some participating players out of the auction.

In a given price range,
\begin{equation}\small
\begin{split}
&\text{if  } \frac{\partial U_i}{\partial a_i}>0, \quad \frac{U(a_i^{(n+1)},\boldsymbol{a}_{-i}^{(n)})-U(a_i^{(n)},\boldsymbol{a}_{-i}^{(n)})}{a_i^{(n+1)}-a_i^{(n)}} >0,\\
&\text{if  } \frac{\partial U_i}{\partial a_i}<0, \quad \frac{U(a_i^{(n+1)},\boldsymbol{a}_{-i}^{(n)})-U(a_i^{(n)},\boldsymbol{a}_{-i}^{(n)})}{a_i^{(n+1)}-a_i^{(n)}} <0.\\
\end{split}
\end{equation}
An increasing/decreasing monotonic function would approach to its upper/lower boundaries. It is not difficult to obtain the upper boundary from concavity:
\begin{equation}\small
U(a_i^{(n+1)},\boldsymbol{a}_{-i}^{(n)}) \le U(a_i^{(n)},\boldsymbol{a}_{-i}^{(n)})+ \frac{\partial U_i}{\partial a_i}(a_i^{(n+1)}-a_i^{(n)}).
\end{equation}
For the lower boundary,
\begin{equation}\small
\begin{split}
&U(a_i^{(n+1)},\boldsymbol{a}_{-i}^{(n)}) \\
=& U(a_i^{(n)},\boldsymbol{a}_{-i}^{(n)})+ \frac{\partial U_i}{\partial a_i}(a_i^{(n+1)}-a_i^{(n)})\\
&-\int_0^{a_i^{(n+1)}-a_i^{(n)}} \frac{\partial U_i(a_i^{(n)},\boldsymbol{a}_{-i}^{(n)})}{\partial a_i^{(n)}}t-\frac{\partial U_i(a_i^{(n)}+t,\boldsymbol{a}_{-i}^{(n)})}{\partial (a_i^{(n)}+t)} t \,\mathrm{d}t.
\end{split}
\end{equation}
In particular, because $\frac{\partial U_i}{\partial a_i}$ is Lipschitz continuous when the weight holds the price in a range,
\begin{equation}\small
||\frac{\partial U_i(a_i^{(n)},\boldsymbol{a}_{-i}^{(n)})}{\partial a_i^{(n)}}-\frac{\partial U_i(a_i^{(n)}+t,\boldsymbol{a}_{-i}^{(n)})}{\partial (a_i^{(n)}+t)}|| \le  L ||a_i^{(n)}-(a_i^{(n)}+t)||.
\end{equation}
Thus, we have
\begin{equation}\small
\begin{split}
U(a_i^{(n+1)},\boldsymbol{a}_{-i}^{(n)}) \ge &U(a_i^{(n)},\boldsymbol{a}_{-i}^{(n)})+ \frac{\partial U_i}{\partial a_i}(a_i^{(n+1)}-a_i^{(n)})\\
&-\frac{1}{2}L(a_i^{(n+1)}-a_i^{(n)})^2.
\end{split}
\end{equation}
\remark The upper and lower boundary of $U(a_i^{(n+1)},\boldsymbol{a}_{-i}^{(n)})$ are also bounded by (\ref{eq:brrange2}).

Due to the above-mentioned price and boundary analysis, we can conclude that $|a_i^{(n+1)}-a_i^{(n)}|<\varepsilon$ after some iterations $\gamma$, where $\varepsilon$ is a small value. Substituting in (\ref{eq:al}), we obtain (at the final iteration):
\begin{equation}\small
\begin{split}
(1-w)a_i^{(T_f+1)}=&(1-w) r(\boldsymbol{a}_{-i}^{(T_f)})\\
a_i^{(T_f+1)}=&r(\boldsymbol{a}_{-i}^{(T_f)}),
\end{split}
\end{equation}
which is the \emph{best response} of $a_i^{(n)}$. Consequently, our algorithm converges to an NE.
\end{proof}

\begin{table}[!t]\vspace{-0.2cm}
%\scriptsize
  \centering
  \caption{%\mycaption{%\vspace*{-1em}
    \vspace*{-0.4em}Proposed Energy Trading Solution}\vspace*{-1em}
    \begin{tabular}{p{8cm}}
      \hline
\textbf{Phase 1 - Proposed Dynamics:}   \vspace*{.1em}\\
\hspace*{1em}Each storage unit $i \in \mathcal{N}$ chooses a starting strategy $a_i^{\textrm{init}}=B_i$\vspace*{.1em}\\
\hspace*{2em}\textbf{repeat,}\vspace*{.2em}\\
\hspace*{1em}a) Each seller $i \in \mathcal{N}$ observes its best response strategy $r_i(\boldsymbol{a}_{-i})$\vspace*{.1em}\\
\hspace*{1em}b) Each seller $i \in \mathcal{N}$ randomly selects the better response strategy\\
\hspace*{1em}between the current strategy and best response strategy in (\ref{eq:al}):\\
\hspace*{1em} $wa_i+(1-w) r_i(\boldsymbol{a}_{-i})$ , where $0 \le w \le 1$. As using the method:\vspace*{.1em}\\
\hspace*{3em}a) An auctioneer (utility operator) communicates with\vspace*{.1em}\\
\hspace*{3em}the buyers and sellers using the grid's two-way communication\vspace*{.1em}\\
\hspace*{3em}architecture (see \cite{EW03} or \cite{hossain2012smart} and references therein).\vspace*{.1em}\\
\hspace*{3em}b) The price and amounts of energy to be traded are\vspace*{.1em}\\
\hspace*{3em}found via the double auction of Section~\ref{sec:da}.\vspace*{.1em}\\
\hspace*{4em}\textbf{Auction}\vspace*{.1em}\\
\hspace*{4em}a) The auctioneer advertises $s_i,\ \forall i \in \mathcal{N}$ and $b_k\ \forall k \in \mathcal{K}$.\vspace*{.1em}\\
\hspace*{4em}b) Each seller publishes its expected price, and the\\
\hspace*{4em}auctioneer orders the sellers as required.\vspace*{.1em}\\
\hspace*{4em}c) After ordering, the auctioneer tells seller $i$, during\\
\hspace*{4em}its turn, of the current vector of strategies $\boldsymbol{a}_{-i}$.\vspace*{.1em}\\
\hspace*{4em}d) Seller $i$ computes and submits its strategic response in\vspace*{.1em}\\
\hspace*{4em}(\ref{eq:al}) using (\ref{eq:br}). \vspace*{.1em}\\
\hspace*{2em}\textbf{until} convergence to an NE strategy vector $\boldsymbol{a}^*$.\vspace*{.2em}\\
\textbf{Phase 2 - Market and Trading}   \vspace*{.1em}\\
\hspace*{1em}a) The auctioneer performs the double auction mechanism\vspace*{.1em}\\
\hspace*{1em}given the equilibrium choices as per $\boldsymbol{a}^*$.\vspace*{.1em}\\
\hspace*{1em}b) Actual energy exchange occurs and revenues are collected.\vspace*{.1em}\\
   \hline
    \end{tabular}\label{tab:algo}\vspace{-0.8cm}
\end{table}

Determining the precise computational complexity for the proposed approach is challenging due to the fact that the trading price varies during the iterative process, which subsequently leads to a varying number of participating sellers. However, we can obtain some insights on the computational complexity by assuming a constant trading price $\bar p(\boldsymbol{a})$ in (19). The computational complexity for comparing the amount sold by sellers with the energy requested by the buyers in (19) is $O(L+M)$, where seller $L$ and buyer $M$ determine the trading price (see Fig.~\ref{fig:intersect}). In this respect, if assume that the trading price does not change, the amount of computation required to calculate $\frac{\bar{p}-s_i}{2\tau_i}$ in (19) is $O(1)$ since this value is independent of the market size. The computational complexity needed for calculating $\frac{(\bar{p}-s_i)(L-1)+2\tau_i(\sum_{j=1,j\neq i}^{L-1}a_j-\sum_{k=1}^{M-1}x_k)}{2\tau_i(L-2)}$ in (19) is $O(L+M)$. This represents the individual seller's computational complexity for the proposed sequential algorithm. As we have $(L-1)$ participating sellers, if the trading price is a constant, the total computational complexity of the proposed sequential algorithm is $O\biggl((L-1)(L+M)\biggr)$. For the parallel algorithm, the proposed approach would require a lower computational complexity, $O(L+M)$, but then it will lead to more iterations than in the sequential case as seen from Fig.~\ref{fig:seqparconv}.

The sellers and buyers in the proposed noncooperative game can interact using a novel algorithm composed of two phases: a strategic dynamics phase and an actual market and energy trading stage. Our strategic dynamics stage begins with every seller choosing an initial strategy $a_i^{\textrm{init}}$. While this initial strategy can be chosen arbitrarily by each seller, the most intuitive choice is that each seller starts by trying to sell all of its available surplus of stored energy. Therefore, we let $a_i^{\textrm{init}}=B_i,\ \forall i \in \mathcal{N}$. Subsequently, an iterative process begins in which the sellers can take turns in choosing their maximum amount of energy to sell (i.e., strategies). To this end, at any iteration $\theta$, any storage unit $i$, during its turn to act, will choose a strategy that approaches its best response strategy, as given by (\ref{eq:br}) and (\ref{eq:al}). This iterative algorithm is executed until guaranteed convergence to an NE. In particular, the proposed algorithm has been shown to always converge to a Nash equilibrium in Theorem \ref{th:converge}. A summary of the proposed algorithm is given in Table~\ref{tab:algo}. Here, we note that, although in Table~\ref{tab:algo} we present a sequential implementation of the algorithm, the players may also utilize a parallel approach. In a sequential implementation, the players act sequentially, in a arbitrary order such that each player is able to observe (or is notified by the auctioneer) the actions taken by the previous players. In contrast, in a parallel approach, at an iteration $t$, all players respond, using (\ref{eq:al}), to the actions observed by the other players at iteration $t-1$. Once an NE of the game is reached, the last phase of the algorithm is the practical market operation. During this final phase, given the equilibrium strategies, all sellers and buyers submit their bids and then engage to an actual double auction in which each storage unit discharges (sells) the desired energy amount and is rewarded accordingly.

One possible approach to determine $w$ is to gradually lower its value from $1$ until a suitable value guarantees reaching an NE. This weight essentially ensures that the price, which introduces the discontinuity, is bounded within a certain range after a certain number of iterations. Determining this weight depends on the various buyers/sellers parameters. The control center of the utility company needs to dynamically determine a weight and change it if needed using time-dependent information observed from the participating users. At the beginning, the control center could start with an initial weight (either arbitrarily or chosen based on historical data), and set a periodic time to dynamically adjust the weight. The control center can gradually and periodically optimize the current weight to meet the observed network environment. One possible strategy is to follow a classical bisection method ~\cite{BO00}. The application of this method helps the control center adjust the speed of convergence. At each step, the center divides the weight interval into two. A subinterval is selected depending on whether an NE could be effectively obtained or not. The center would use a previous subinterval as a new interval in the next step and this process is continued until the interval is sufficiently small and convergence is guaranteed. Nonetheless, we have to stress that we have proven that there exists a weight $w,\  0<w<1$ such that the proposed algorithm is guaranteed to converge to an NE, and thus, in practice, control center will eventually reach this weight, as the range within which the weight varies is finite.

For practical implementation, the utility company's control center acts as an \emph{auctioneer}~\cite{DR00}  that guides the energy market and monitors the interactions of the storage units. This auctioneer utilizes a storage-to-grid communication network such as in \cite{EW03} or \cite{hossain2012smart} (and references therein) to communicate with all grid elements and storage units. Thus, the auctioneer plays mainly two roles: (a) sorting out the sellers
and the buyers once bids are received, as per (\ref{eq:sel}) and (\ref{eq:buy}) and (b) gathering their strategies whenever they must act during the dynamics phase.  At a given iteration $t$ during Phase I of the proposed approach, any seller must compute its utility in (\ref{eq:util1}) so as to update its strategy. The strategy update can be based either on the current state (for the sequential algorithm) or on the state of the grid in the previous iteration (for the parallel implementation). For enabling this strategy update, the auctioneer and the storage units interact over the communication infrastructure using either an open or a private method. In the open method, when interacting with a certain seller, the auctioneer conveys the current opponents strategy vector $\boldsymbol{a}_{-i}$, the reservation bids, and the type of auction being
implemented. Once this information is obtained, the seller can calculate its best strategy using classical optimization techniques [31] and pursue the proposed auction procedure.

The disadvantage of the open method is that the auctioneer must disclose the current strategies and reservation prices/bids to the sellers. In many cases, it is of interest to keep this information private. Hence, alternatively, when an auctioneer
communicates with a certain seller $i$, this seller  will submit a restricted set of potential strategies $\bar{\mathcal{A}}_i\subset\mathcal{A}_i$ for the current iteration. Then, the control center feeds back the trading prices and amount of energy that seller $i$ will potentially obtain
at the current time, for each of the submitted strategies. Using this data, each seller $i$ builds, using function
smoothing methods such as in \cite{SMOOTH}, an estimate of its utility under the current strategies, and, subsequently, find the optimal strategy response for this iteration using optimization methods. For this private implementation, the sellers do not require any knowledge on the type of auction being used nor on the strategies and reservation prices of their opponents (or the bids of the buyers). In addition, the control center does not require any knowledge of the utility functions that are being used by the players. The only information that would circulate, as dictated by the method, would be the trading price and the energy sold for any potential strategy $a_i$ submitted by any seller $i$ to the auctioneer during its turn.

We finally note that, in the presence of an elaborate communication infrastructure, the sellers and the buyers can interact directly without the need for a control center. In this case, each seller can, individually, decide on the amount of energy it wants to sell at each iteration of the proposed approach, while directly notifying the other players of its choice. The rest of the operation would still follow the iterative process discussed earlier.\vspace{-0.2cm} 

\section{Simulation Results and Analysis}\label{sec:sim}\vspace{-0.2cm}
For simulating the proposed system, we consider a geographical region in which a number of storage units have a surplus of stored energy that they wish to sell to existing customers (e.g. loads, substations, etc.) in a smart grid. Each unit has a surplus between $75$~MWh and $220$MWh that can be sold. The reservation prices of the sellers are chosen randomly from a range of $[10,50]$ dollars per MWh while reservation bids of the buyers are chosen randomly from a range of $[15,60]$ dollars per MWh. The demand of each buyer is chosen randomly from within a range of $[20,60]$~MWh. Unless stated otherwise, the cost per energy sold is set to $\tau_i = 0.5,\ \forall i\in \mathcal{N}$. All statistical results are averaged over all possible random values for the different parameters (prices, bids, demand, etc.) using a large number of independent simulation runs.

\begin{figure}[!t]
  \begin{center}
   \vspace{-0.2cm}
    \includegraphics[width=9cm]{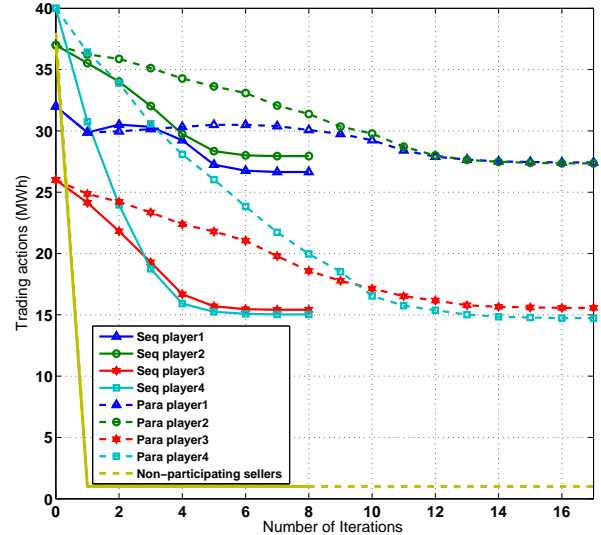}
    \vspace{-0.5cm}
    \caption{\label{fig:seqparconv}Average action per seller (storage unit) resulting from the proposed game approach and from the number of storage units $K=5$~buyers, $N=6$~sellers.}
 \end{center}\vspace{-0.3cm}
\end{figure}

\begin{figure}[!t]
  \begin{center}
   \vspace{-0.2cm}
    \includegraphics[width=7cm]{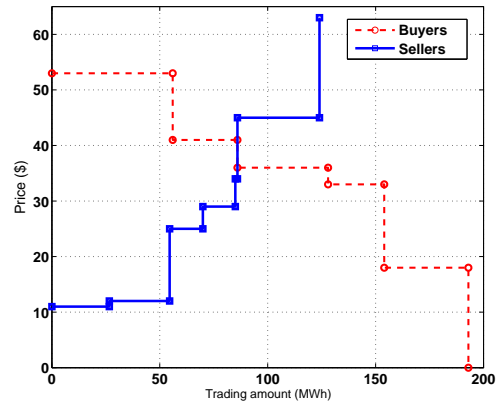}
    \vspace{-0.42cm}
    \caption{\label{fig:seqPAend} The ending double auction market performance of storage units $K=5$~buyers, $N=6$ based on sequential iterations.}
 \end{center}\vspace{-1.1cm}
\end{figure}

In Fig.~\ref{fig:seqparconv}, we show, for a smart grid with $K=5$~buyers and $N=6$~sellers, the average action per seller (storage unit) resulting from the proposed game approach at the equilibrium using both sequential and parallel approaches. The sequential algorithm's performance is compared with that of the parallel algorithm in which the sellers, simultaneously, attempt to sell energy depending on their previous actions. Here, in particular, we choose the same weight $w=0.3$ for both the sequential and the parallel algorithms. In Fig.~\ref{fig:seqparconv}, we can see that, for the proposed algorithm, the trading action per storage unit converges to different values with increasing iterative steps. Fig.~\ref{fig:seqparconv} shows that two players out of six decide not to participate in the market. The brown line relates to those sellers who do not participate in the market due to the fact that would the trading price would then lead to a negative utility. However, these sellers would still have some energy in hand and they can offer it for sale at a later time instant in which the trading price in another auction or area might give them an opportunity to obtain positive utility. Thus, although they do not trade at the current market price, they will maintain their available energy and eventually participate in a future market. In particular, we use the ``brown line'' to indicate a baseline action value of $1$ to represent those sellers that do not participate in the market, but rather prefer to wait for future trading opportunities. We can also observe that, for the sequential algorithm (solid line), the action of player 1 increases a little at the beginning. This is due to the fact that, the player who plays first in one iteration has a higher opportunity to sell energy than others. In general, as seen in Fig.~\ref{fig:seqparconv}, because of the competition over the resources, the actions are essentially decreasing, which means that at the equilibrium, not all players will sell their maximum available energy.

Given the same setting as in Fig.~\ref{fig:seqparconv}, Fig.~\ref{fig:seqPAend} shows the price resulting from the double auction phase for the case of the sequential algorithm. In this figure, the intersection point demonstrates that seller $5$ and buyer $2$ determine the trading price. The total amount sold by participating sellers (seller 1 to 4) is equal to the demand of participating buyers (buyer 1 and 2). If the solid and dashed lines intersect at a point of the $3$rd buyer (the $3$rd range of the dashed line), seller $5$ and buyer $3$ lower the trading price. Although all four participating sellers might sell more than before, the associated reduction in the trading price will lead to lower revenues.
\begin{figure}[!t]
  \begin{center}
   \vspace{-0.25cm}
    \includegraphics[width=7cm]{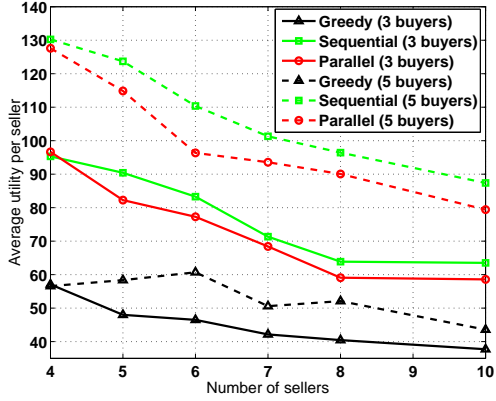}
    \vspace{-0.42cm}
    \caption{\label{fig:utility} Performance assessment in terms of average utility per seller as the number of storage units $N$ varies for $K=3, K=5$~buyers.}
  \end{center}\vspace{-0.95cm}
\end{figure}

Fig.~\ref{fig:utility} presents, for a smart grid with $K=3,5$~buyers, the average achieved utility per seller resulting from the proposed game as the number of storage units $N$ varies. Here, we set $w=0.5$ for the sequential algorithm and $w=0.1$ for the parallel algorithm. For comparison purposes, we develop a conventional, baseline greedy algorithm using which, iteratively, each seller tries to sell the maximum amount that it could sell (while accounting for the changes of the utility in (\ref{eq:util})) while first picking the highest-bid buyers. The greedy process continues until no additional energy trade is possible. In this greedy scheme, the trading prices are selected as the middle point between the concerned buyer's reservation bid and the concerned seller's reservation price. In Fig.~\ref{fig:utility}, we can see that the average utility per storage unit is decreasing with $N$. The reason behind it mainly involves two aspects. First, the increase in sellers can lead to an increased competition and, thus, a decrease in the overall trading price. Second, the number of sellers $L-1 < N$ that will actually participate in the final energy exchange market reaches a certain maximum that no longer increases with $N$ due to the fixed demand (i.e., the number of buyers). This figure demonstrates that, at all $N$, the proposed noncooperative game approach yields a significant performance improvement, in terms of the average utility achieved per storage unit. In particular, this advantage of the proposed approach reaches up to $130.2\%$ (the maximum at $K=5, N=4$) relative to a conventional greedy approach.

Fig.~\ref{fig:iteration} shows, for $K=5$~buyers, the average number of iterations needed before convergence of the different algorithms as the number of storage units $N$ increases. In this figure, we can see that the average number of iterations of the proposed sequential algorithm is similar to that of the classical best response algorithm (whenever this algorithm converges, recall from Theorem~\ref{th:converge} that a best response dynamics may not converge). As expected, Fig.~\ref{fig:iteration} shows that the parallel algorithm requires a much higher number of iterations. In particular, the average number of iterations resulting from the sequential algorithm varies from $7.7$ at $N=6$ to $8.2$ at $N=7$, in contrast, for the parallel case, it varies from $25.8$ at $N=4$ to $32.5$ at $N=10$. This result indicates that the proposed algorithm, particularly with a sequential implementation, has a reasonably fast convergence speed.

\begin{figure}[!t]
  \begin{center}
   \vspace{-0.25cm}
    \includegraphics[width=7cm]{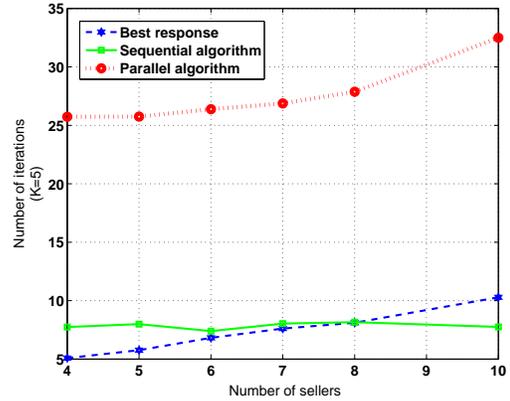}
    \vspace{-0.4cm}
    \caption{\label{fig:iteration} Average number of iterations per seller as the number of storage units $N$ varies for $K=5$~buyers.}
 \end{center}\vspace{-0.26cm}
\end{figure}

\begin{figure}[!t]
  \begin{center}
   \vspace{-0.35cm}
    \includegraphics[width=7cm]{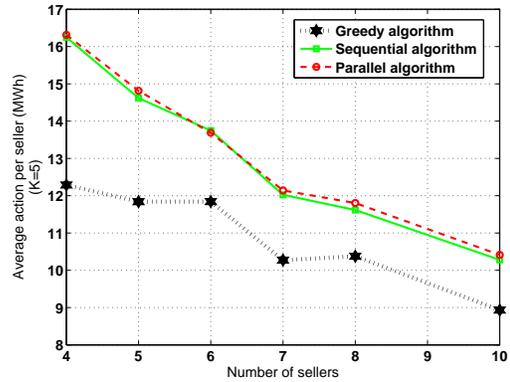}
    \vspace{-0.4cm}
    \caption{\label{fig:action} Average number of action per seller resulting from the proposed game approach as the number of storage units $N$ varies for $K=5$~buyers.}
  \end{center}\vspace{-0.94cm}
\end{figure}
Fig.~\ref{fig:action} shows for $K=5$~buyers, the average action per seller for both the sequential and the parallel algorithms as the number of storage units $N$ grows. We can see that the average action per player resulting from both the sequential and the parallel algorithms is greater than that of the greedy strategy. Thus, using the proposed algorithm provides the sellers with more incentives to trade larger amounts in the markets. This is in fact further reflected in the enhanced utility achieved by the proposed approach, as seen in Fig.~\ref{fig:utility}. Finally, Fig.~\ref{fig:action} also shows that the sequential and the parallel algorithms converge to nearly the same actions at the equilibrium.

Fig.~\ref{fig:tau} shows, for different buyers and sellers, the average utility per seller as the penalty factor $\tau$ varies. From (\ref{eq:util}), we can see that the utility would decrease with increasing $\tau$ and this is corroborated in Fig.~\ref{fig:tau}. In particular, when $\tau$ is equal to $1$, the utility of each player is dramatically influenced by the penalty part in (\ref{eq:util}). The first revenue term in (\ref{eq:util}) the sellers obtained in the auction remains the same as the second penalty term increases, even though the total utility is positive.

Fig.~\ref{fig:difbuyer} shows the average utility from the different proposed approaches as the number of buyers $K$ varies, for $N=6$ sellers. Each iteration consists of a series of choices by the sellers, using the same initial information. In Fig.~\ref{fig:difbuyer}, we can also see that, as the number of buyers, $K$, increases the average utility per seller increases due to the availability of additional buyers that are willing to participate in the market. In fact, Fig.~\ref{fig:difbuyer} shows that, as $K$ increases, the sellers have a larger utility due to the availability of more buyers. In particular, our proposed algorithm yields a performance improvement ranging between $72.3\%$ (for $K=4, N=6$) to $234.4\%$ (for $K=10, N=6$) relative to the greedy scheme. Further inspection of Fig.~\ref{fig:difbuyer} reveals that any change in the numbers of buyers does not impact the increasing average utility rate of our algorithm, while the greedy algorithm reaches a maximum when the number of buyers are similar to that of sellers.

\begin{figure}[!t]
  \begin{center}
   \vspace{-0.2cm}
    \includegraphics[width=7cm]{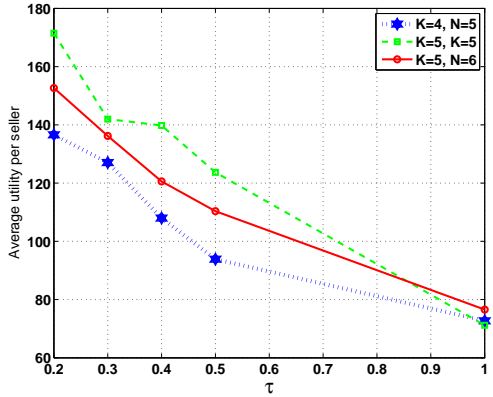}
    \vspace{-0.45cm}
    \caption{\label{fig:tau} Performance assessment in terms of penalty factor $\tau$ resulting from the proposed game approach for different buyers and sellers.}
  \end{center}\vspace{-0.5cm}
\end{figure}

\begin{figure}[!t]
  \begin{center}
   \vspace{-0cm}
    \includegraphics[width=7cm]{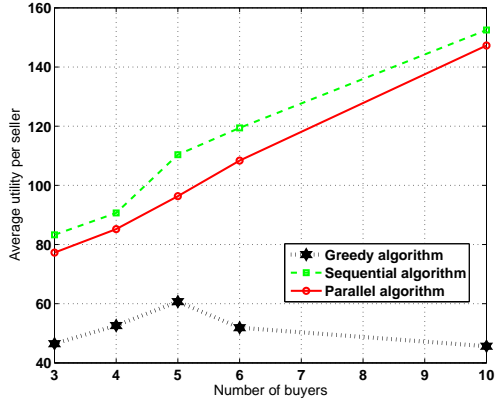}
    \vspace{-0.4cm}
    \caption{\label{fig:difbuyer} Performance assessment in terms of average utility per seller as the number of buyers $K$ varies, $N=6$~sellers.}
  \end{center}\vspace{-0.8cm}
\end{figure}

Fig.~\ref{fig:stateschange} shows, for a smart grid with $K=3$~buyers and $N=4$~sellers, the state of charge represented by the battery amount per player resulting from a time-dependent game as time evolves. Here, we show the results for the proposed sequential algorithm with a weight $w=0.5$, without loss of generality. For our simulations, we assume that the period corresponds to $1$ hour as is typical in a residential community ~\cite{shao2010impact}. In Fig.~\ref{fig:stateschange}, we can see that, for the proposed game, four sellers would sell their energy in the market during the first time instant. Then, after the first run, all players reconsider their roles and still participate in the market, especially for those who did not sell/buy enough in the previous time slots. The iteration would then lead to a new trading price and this process continues, as previous users are still in the market and no new players join in this group. In Fig.~\ref{fig:stateschange}, we can see that all players have an opportunity to act as sellers or buyers every hour. Player $1$, for example, acts as a buyer during the first hour. After the second hour, this player changes from a buyer to a seller, and then acts as a buyer at the fifth hour. Player $4$ sells a small amount in the first three hours and becomes a seller at the fourth and sixth time instants. During the second and third hours, Player $2$ does not sell a large amount despite the fact that it had acted as a buyer and fully charged at the first time instant. After $4$ hours have elapsed, Player 2 sells $12.8$ MWh. This player tends to sell the energy because it reaches its battery limitation and has an incentive to become a seller after the first hour. 

\begin{figure}[!t]
  \begin{center}
   \vspace{-0.2cm}
    \includegraphics[width=7cm,height=5.5cm]{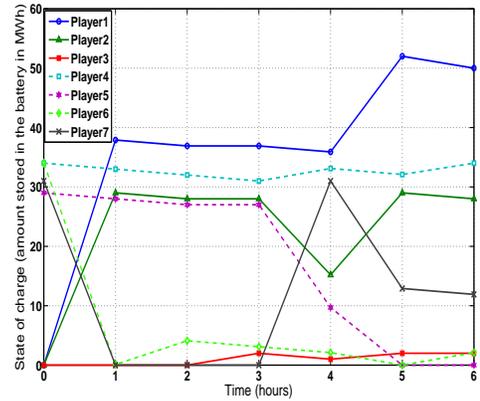}
    \vspace{-0.5cm}
    \caption{\label{fig:stateschange} Performance assessment in terms of the amount per player as the number of runs varies for initial $K=3$~buyers and $N=4$~sellers.}
  \end{center}\vspace{-0.5cm}
\end{figure}

\begin{figure}[!t]
  \begin{center}
   \vspace{-0cm}
    \includegraphics[width=7cm,height=5.4cm]{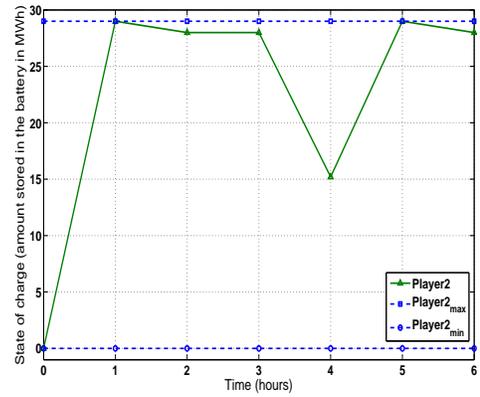}
    \vspace{-0.4cm}
    \caption{\label{fig:Bmax} Performance assessment in terms of the amount of Player $2$ as the number of runs varies for comparing with underlying batter limitation.}
  \end{center}\vspace{-0.8cm}
\end{figure}

In Fig.~\ref{fig:Bmax}, we show how the amount of energy in Player 2's battery changes within this player's minimum and maximum battery capacity. In this figure, we can see that Player $2$ reaches its maximum battery size capacity two times. Correspondingly, this player acts as a seller twice at the second and fourth hour. In contrast to Players 3 and 4, Player $2$ frequently uses its storage unit so as to obtain potential utility through time-dependently buying/selling energy. In essence, Fig.~\ref{fig:Bmax} shows how the proposed game can be used to handle the battery limitations of the users as well as their time-dependent behavior.\vspace{-0.2cm}

\section{Conclusions}\label{sec:conc}\vspace{-0.2cm}
In this paper, we have introduced a novel approach for studying the complex interactions between a number of storage units seeking to sell part of their stored energy surplus to smart grid elements. We have formulated a noncooperative game between the storage units in which each unit strategically chooses the maximum amount of energy surplus that it is willing to sell so as to optimize a utility function that captures the benefits from energy selling as well as the associated costs. To determine the trading price that governs the energy trade market between storage units and smart grid elements, we have proposed an approach based on double auctions, which leads to a strategy-proof outcome. We have shown the existence of a Nash equilibrium and studied its properties. Further, to solve the underlying game, we have proposed a novel algorithm using which the storage units can reach a Nash equilibrium point for our model. Simulation results have shown that the proposed approach enables the storage units to act strategically while improving their average utility. For future work, it is  of interest to extend the model to a dynamic game model in which all players could  time-dependently observe each others' strategies as well as the grid's state and dynamically determine their underlying actions. In this respect, the work done in this paper serves as a basis for developing such a more elaborate dynamic game model in which players can make long-term decisions with regard to their energy trading processes.

\def\baselinestretch{0.9}
\bibliographystyle{IEEEtran}
\bibliography{references}

\begin{thebibliography}{10}
\providecommand{\url}[1]{#1}
\csname url@samestyle\endcsname
\providecommand{\newblock}{\relax}
\providecommand{\bibinfo}[2]{#2}
\providecommand{\BIBentrySTDinterwordspacing}{\spaceskip=0pt\relax}
\providecommand{\BIBentryALTinterwordstretchfactor}{4}
\providecommand{\BIBentryALTinterwordspacing}{\spaceskip=\fontdimen2\font plus
\BIBentryALTinterwordstretchfactor\fontdimen3\font minus
  \fontdimen4\font\relax}
\providecommand{\BIBforeignlanguage}[2]{{%
\expandafter\ifx\csname l@#1\endcsname\relax
\typeout{** WARNING: IEEEtran.bst: No hyphenation pattern has been}%
\typeout{** loaded for the language `#1'. Using the pattern for}%
\typeout{** the default language instead.}%
\else
\language=\csname l@#1\endcsname
\fi
#2}}
\providecommand{\BIBdecl}{\relax}
\BIBdecl

\bibitem{farhangi2010path}
H.~Farhangi, ``The path of the smart grid,'' \emph{IEEE Power and Energy
  Magazine}, vol.~8, no.~1, pp. 18--28, Jan. 2010.

\bibitem{kazempour2009electric}
S.~J. Kazempour, M.~P. Moghaddam, M.~Haghifam, and G.~Yousefi, ``Electric
  energy storage systems in a market-based economy: Comparison of emerging and
  traditional technologies,'' \emph{Elsevier Renewable Energy}, vol.~34,
  no.~12, pp. 2630--2639, Jun. 2009.

\bibitem{hadjipaschalis2009overview}
I.~Hadjipaschalis, A.~Poullikkas, and V.~Efthimiou, ``Overview of current and
  future energy storage technologies for electric power applications,''
  \emph{Elsevier Renewable and Sustainable Energy Reviews}, vol.~13, no.~6, pp.
  1513--1522, Aug. 2009.

\bibitem{HosseinAkhavanHejazi}
H.~Akhavan-Hejazi and H.~Mohsenian-Rad, ``A stochastic programming framework
  for optimal storage bidding in energy and reserve markets,'' in \emph{Proc.
  of the IEEE PES Conference on Innovative Smart Grid Technologies
  (ISGT’13)}, Washington, DC, Feb. 2013.

\bibitem{aguero2012integration}
J.~R. Aguero, P.~Chongfuangprinya, S.~Shao, L.~Xu, F.~Jahanbakhsh, and H.~L.
  Willis, ``Integration of plug-in electric vehicles and distributed energy
  resources on power distribution systems,'' in \emph{Electric Vehicle
  Conference (IEVC)}, Gaithersburg, MD, USA, Mar. 2012.

\bibitem{hossain2012smart}
E.~Hossain, Z.~Han, and H.~V. Poor, \emph{Smart Grid Communications and
  Networking}.\hskip 1em plus 0.5em minus 0.4em\relax Cambridge University
  Press, 2012.

\bibitem{lassila2012methodology}
J.~Lassila, J.~Haakana, V.~Tikka, and J.~Partanen, ``Methodology to analyze the
  economic effects of electric cars as energy storages,'' \emph{IEEE
  Transactions on Smart Grid}, vol.~3, no.~1, pp. 506--516, Mar. 2012.

\bibitem{hatziargyriou2007overview}
N.~Hatziargyriou, H.~Asano, R.~Iravani, and C.~Marnay, ``An overview of ongoing
  research, development, and demonstration projects,'' \emph{IEEE Power \&
  Energy}, vol.~8, pp. 78--94, 2007.

\bibitem{thatte2012towards}
A.~Thatte and L.~Xie, ``Towards a unified operational value index of energy
  storage in smart grid environment,'' \emph{IEEE Transactions on Smart Grid},
  vol.~3, no.~3, pp. 1418--1426, Sep. 2012.

\bibitem{EW07}
N.~Rotering and M.~Ilic, ``Optimal plug-in electric vehicle charge control in
  deregulated electricity markets,'' \emph{IEEE Transactions on Power Systems},
  Aug 2011.

\bibitem{lindley2010smart}
D.~Lindley, ``Smart grids: The energy storage problem,'' \emph{Nature}, vol.
  463, no. 7277, p.~18, Jan. 2010.

\bibitem{rastler2010electricity}
D.~Rastler, \emph{Electricity Energy Storage Technology Options: A White Paper
  Primer on Applications, Costs and Benefits}.\hskip 1em plus 0.5em minus
  0.4em\relax Electric Power Research Institute, 2010.

\bibitem{lu2004pumped}
N.~Lu, J.~H. Chow, and A.~A. Desrochers, ``Pumped-storage hydro-turbine bidding
  strategies in a competitive electricity market,'' \emph{IEEE Transactions on
  Power Systems}, vol.~19, no.~2, pp. 834--841, May 2004.

\bibitem{garcia2008stochastic}
J.~Garcia-Gonzalez, R.~R. de~la Muela, L.~M. Santos, and A.~M. Gonz{\'a}lez,
  ``Stochastic joint optimization of wind generation and pumped-storage units
  in an electricity market,'' \emph{IEEE Transactions on Power Systems},
  vol.~23, no.~2, pp. 460--468, May 2008.

\bibitem{sioshansi2009estimating}
R.~Sioshansi, P.~Denholm, T.~Jenkin, and J.~Weiss, ``Estimating the value of
  electricity storage in pjm: Arbitrage and some welfare effects,''
  \emph{Elsevier Energy economics}, vol.~31, no.~2, pp. 269--277, Aug. 2009.

\bibitem{caralis2012role}
G.~Caralis, D.~Papantonis, and A.~Zervos, ``The role of pumped storage systems
  towards the large scale wind integration in the greek power supply system,''
  \emph{Elsevier Renewable and Sustainable Energy Reviews}, vol.~16, no.~5, pp.
  2558--2565, Jun. 2012.

\bibitem{diaz2012review}
F.~D{\'\i}az-Gonz{\'a}lez, A.~Sumper, O.~Gomis-Bellmunt, and
  Villaf{\'a}fila-Robles, ``A review of energy storage technologies for wind
  power applications,'' \emph{Elsevier Renewable and Sustainable Energy
  Reviews}, vol.~16, no.~4, pp. 2154--2171, May 2012.

\bibitem{garnier2009integrated}
C.~Garnier, J.~Currie, and T.~Muneer, ``Integrated collector storage solar
  water heater: Temperature stratification,'' \emph{Elsevier Applied Energy},
  vol.~86, no.~9, pp. 1465--1469, Sep. 2009.

\bibitem{weaver2009game}
W.~W. Weaver and P.~T. Krein, ``Game-theoretic control of small-scale power
  systems,'' \emph{IEEE Transactions on Power Delivery}, vol.~24, no.~3, pp.
  1560--1567, 2009.

\bibitem{EW02}
E.~Baeyens, E.~Y. Bitar, P.~P. Khargonekar, and K.~Poolla, ``Wind energy
  aggregation: A coalitional game approach,'' in \emph{Proc. 50th IEEE
  Conference on Decision and Control (CDC)}, Orlando, FL, USA, Dec. 2011.

\bibitem{xu2010research}
X.~Z. Xu and Q.~Chen, ``Research on reactive power compensation algorithm based
  on game theory in wind farm,'' \emph{Key Engineering Materials}, vol. 439,
  pp. 989--993, 2010.

\bibitem{hobbs2000strategic}
B.~F. Hobbs, C.~B. Metzler, and J.-S. Pang, ``Strategic gaming analysis for
  electric power systems: An mpec approach,'' \emph{IEEE Transactions on Power
  Systems}, vol.~15, no.~2, pp. 638--645, 2000.

\bibitem{ng2006game}
S.~K. Ng, C.~Lee, and J.~Zhong, ``A game-theoretic approach to study strategic
  interaction between transmission and generation expansion planning,'' in
  \emph{North American Power Symposium, 38th}, North American, 2006, pp.
  115--12.

\bibitem{bompard2006network}
E.~Bompard, W.~Lu, and R.~Napoli, ``Network constraint impacts on the
  competitive electricity markets under supply-side strategic bidding,''
  \emph{IEEE Transactions on Power Systems}, vol.~21, no.~1, pp. 160--170, Feb.
  2006.

\bibitem{PH01}
T.~H. Bradley and A.~A. Frank, ``Design, demonstrations and sustainability
  impact assessments for plug-in hybrid electric vehicles,'' \emph{Elsevier
  Renewable and Sustainable Energy Reviews}, vol.~13, pp. 115--128, Jan. 2009.

\bibitem{PH00}
A.~Simpson, ``Cost-benefit analysis of plug-in hybrid electric vehicle
  technology,'' in \emph{Proc. 22nd International Battery, Hybrid and Fuel Cell
  Electric Vehicle Symposium}, Yokohama, Japan, Oct. 2006.

\bibitem{GT00}
T.~Ba\c{s}ar and G.~J. Olsder, \emph{Dynamic Noncooperative Game Theory}.\hskip
  1em plus 0.5em minus 0.4em\relax Philadelphia, PA: SIAM Series in Classics in
  Applied Mathematics, 1999.

\bibitem{DR00}
D.~Friedman, D.~P. Friedman, and J.~Rust, \emph{The Double Auction Market:
  Institutions, Theories, and Evidence}.\hskip 1em plus 0.5em minus 0.4em\relax
  Boulder, CO: Westview Press, 1993.

\bibitem{DR01}
P.~Huang, A.~Scheller-Wolf, and K.~Sycara, ``Design of a multi-unit double
  auction e-market,'' \emph{Computational Intelligence}, vol.~18, no.~4, pp.
  596--617, Feb. 2002.

\bibitem{dasgupta1986existence}
P.~Dasgupta and E.~Maskin, ``The existence of equilibrium in discontinuous
  economic games, i: Theory,'' \emph{JSTOR The Review of Economic Studies},
  vol.~53, no.~1, pp. 1--26, Jan. 1986.

\bibitem{BO00}
S.~Boyd and L.~Vandenberghe, \emph{Convex Optimization}.\hskip 1em plus 0.5em
  minus 0.4em\relax New York, USA: Cambridge University Press, Sep. 2004.

\bibitem{EW03}
B.~Jansen, C.~Binding, O.~Sundstrom, and D.~Gantenbein, ``Architecture and
  communication of an electric vehicle virtual power plant,'' in \emph{Proc.
  International Conference on Smart Grid Communications}, Gaithersburg, MD,
  USA, Oct. 2010.

\bibitem{SMOOTH}
A.~W. Bowman and A.~Azzalini, \emph{Applied Smoothing Techniques for Data
  Analysis: The Kernel Approach with S-Plus Illustrations: The Kernel Approach
  with S-Plus Illustrations}.\hskip 1em plus 0.5em minus 0.4em\relax Oxford
  University Press, 1997.

\bibitem{shao2010impact}
S.~Shao, T.~Zhang, M.~Pipattanasomporn, and S.~Rahman, ``Impact of tou rates on
  distribution load shapes in a smart grid with phev penetration,'' in
  \emph{IEEE PES Transmission and Distribution Conference and Exposition}, Apr.
  2010, pp. 1--6.

\bibitem{saad2011noncooperative}
W.~Saad, Z.~Han, H.~V. Poor, and T.~Ba\c{s}ar, ``A noncooperative game for
  double auction-based energy trading between phevs and distribution grids,''
  in \emph{Proc. IEEE Int. Conf. on Smart Grid Communications (SmartGridComm)},
  Brussels, Belgium, Oct. 2011.

\end{thebibliography}
\nocite{saad2011noncooperative}
\end{document}